\documentclass[aps, pra, superscriptaddress, reprint, floatfix]{revtex4-2}

\usepackage{graphicx}
\usepackage{xcolor}
\usepackage[colorlinks=true, linkcolor=teal, citecolor=teal, urlcolor=teal]{hyperref}

\usepackage{amsmath}
\usepackage{amssymb}
\usepackage{amsfonts}
\usepackage{amsthm}
\usepackage{physics}
\usepackage{booktabs}
\usepackage{array}[=2016-10-06]
\usepackage[capitalize]{cleveref}
\usepackage{etoolbox}
\usepackage{orcidlink}
\usepackage{libertine}

\AtBeginEnvironment{subequations}{%
}

\crefname{equation}{Eq.}{equations}
\crefname{section}{Sec.}{sections}
\crefname{figure}{Fig.}{figures}
\crefname{appendix}{Appendix}{appendices}
\crefname{table}{Table}{tables}

\theoremstyle{plain}
\newtheorem{theorem}{Theorem}
\newtheorem{lemma}[theorem]{Lemma}
\newtheorem{proposition}[theorem]{Proposition}

\newtheorem{corollary}[theorem]{Corollary}
\AfterEndEnvironment{theorem}{\noindent\ignorespaces}
\AfterEndEnvironment{lemma}{\noindent\ignorespaces}
\AfterEndEnvironment{corollary}{\noindent\ignorespaces}

\AfterEndEnvironment{definition}{\noindent\ignorespaces}
\AfterEndEnvironment{proof}{\noindent\ignorespaces}

\crefname{theorem}{Theorem}{theorems}
\crefname{lemma}{Lemma}{lemmas}
\crefname{corollary}{Corollary}{corollaries}
\crefname{definition}{Definition}{definitions}

\newcommand*\diff{\mathop{}\!\mathrm{d}}
\DeclareMathOperator*{\argmin}{arg\,min}

\begin{document}
\title{Projected evolution of quantum gates for circuit simulation with matrix product states}

\author{Aaron Sander\,\orcidlink{0009-0007-9166-6113}}
\affiliation{Technical University of Munich, Munich, Germany}

\author{Maximilian Fröhlich\,\orcidlink{0009-0007-5276-2858}}
\affiliation{Weierstrass Institute for Applied Analysis and Stochastics, Berlin, Germany}

\author{Mazen Ali\,\orcidlink{0000-0003-1664-7098}}
\affiliation{Department of Flow and Material Simulation, Fraunhofer ITWM, Kaiserslautern, Germany}

\author{Martin Eigel\,\orcidlink{0000-0003-2687-4497}}
\affiliation{Weierstrass Institute for Applied Analysis and Stochastics, Berlin, Germany}

\author{Jens Eisert\,\orcidlink{0000-0003-3033-1292}}
\affiliation{Freie Universität Berlin, Berlin, Germany}
\affiliation{Helmholtz-Zentrum Berlin für Materialien und Energie, Berlin, Germany}

\author{Michael Hintermüller\,\orcidlink{0000-0001-9471-2479}}
\affiliation{Weierstrass Institute for Applied Analysis and Stochastics, Berlin, Germany}

\author{Christian B.~Mendl\,\orcidlink{0000-0002-6386-0230}}
\affiliation{Technical University of Munich, Munich, Germany}

\author{Richard M.~Milbradt\,\orcidlink{0000-0001-8630-9356}}
\affiliation{Technical University of Munich, Munich, Germany}

\author{Robert Wille\,\orcidlink{0000-0002-4993-7860}}
\affiliation{Technical University of Munich, Munich, Germany}
\affiliation{MQSC GmbH, Munich, Germany}
\affiliation{Software Competence Center Hagenberg GmbH  (SCCH), Hagenberg, Austria}

\date{\today}

\begin{abstract}
Tensor networks and specifically  
matrix product states (MPS) enable feasible classical  quantum circuit simulation when entanglement remains compressible, but
for this, gate updates must balance accuracy, representation size, and runtime.
Conventional methods incorporate a complete gate into the MPS before compressing enlarged bonds.
Here, we instead use lessons from the time-dependent variational principle (TDVP) to analyze the projected action of each gate on the MPS by reducing it to a continuous generator path.
We establish the spatial and state-space structure of this projected evolution and show that both the fixed-rank and two-site actions are confined to a contiguous window spanned by the gate while retaining boundary information from the surrounding MPS.
Within this window characterization, we delineate when the instantaneous two-site action is exact and when variational directions can block it completely.
Using standard nearest-point-projection geometry, we also show that fixed-rank TDVP and the local nearest fixed-rank projection of the exact gate endpoint agree to first order.
We implement the gate-window restriction with rank-adaptive two-site TDVP and test both the projected update and a hybrid method that reserves it for long-range gates.
Across the circuit benchmarks, these implementations incur runtime overhead and show no consistent accuracy--cost advantage over the tested baseline methods.
The resulting framework provides a gate-specific structural analysis and a tested implementation together with explicit limits on its practical performance.
\end{abstract}

\maketitle

\section{Introduction}
\label{sec:introduction}

Classical simulation of quantum circuits underpins algorithm verification,
device benchmarking, and numerical studies of quantum dynamics.  Exact
state-vector simulation nevertheless requires resources that grow
exponentially with system size, while compressed representations become costly
as the entanglement needed to describe the state increases~\cite{Zhou_2020}.
Recent demonstrations of quantum computational advantage have deliberately
probed regimes near or beyond the reach of available classical methods
~\cite{Arute_2019,Morvan_2024,haghshenas_2026,SupremacyReview}.  These
developments motivate controlled approximations whose errors and computational
limitations can be identified explicitly.

\emph{Matrix product states} (MPS) represent a quantum state as a one-dimensional
chain of tensors~\cite{Schollw_ck_2011,MPSZero},
as a distinctly one-dimensional instance of a tensor network state
\cite{RevModPhys.93.045003,orus2013,AreaReview}.
The dimensions of the connections between neighboring
tensors, called bond dimensions, 
determine both how much bipartite
entanglement the MPS can represent and 
how much storage it requires.  When
moderate bond dimensions suffice, an MPS replaces the exponentially large
state vector by a compact representation whose cost grows only polynomially
with system size~\cite{Schollw_ck_2011,Paeckel_2019}.  Other tensor network
simulators contract the spacetime network of a complete circuit along optimized
paths~\cite{MarkovShi2008,Gray_2021,Gray2024} or use geometries matched to the circuit
connectivity~\cite{PhysRevResearch.6.013326,Tindall_2024} (see Ref.~\cite{Berezutskii2025} for a recent review). Within an MPS
simulator, practical accuracy and cost depend both on the entanglement produced
by the circuit and on how each gate is incorporated without exceeding the
available bond dimensions.

Established methods first incorporate a finite gate and
then reduce the enlarged representation through what we call \emph{apply-and-compress} methods.  For nearest-neighbor gates, direct
tensor contraction followed by a \emph{singular value decomposition} (SVD) gives
the standard \emph{time-evolving block decimation} (TEBD) update
~\cite{Vidal_2003,Vidal_2004,Paeckel_2019}.  Long-range gates can instead be
implemented using SWAP routing~\cite{Zhou_2020}, or applied directly using
a \emph{matrix product operator} (MPO) representing the gate, with the result approximated through sequential contract-and-truncate
steps~\cite{Stoudenmire_2010,Sander_EC_2025}, subsystem-based compression~\cite{tensornetworkorg, CamanoEpperlyTropp2025, Milbradt2026}, or variational
fitting~\cite{Verstraete_2004}.  These constructions
distribute approximation error and computational work differently.  The
distinction becomes especially visible when a gate spans several MPS bonds.

The \emph{time-dependent variational principle} (TDVP) offers a complementary
approximation picture that has become a well-established and powerful framework
for MPS simulations of quantum many-body dynamics
~\cite{Haegeman2011,Haegeman_2013,Haegeman_2016,Paeckel_2019}.
Rather than first constructing an exact time step and compressing it, TDVP keeps
the state within the MPS variational space throughout the evolution.  At each
instant, it projects the exact derivative onto the tangent space of the current
MPS and selects the available direction with the smallest Hilbert-space
residual.  This principle underlies \emph{fixed-rank single-site TDVP} (1TDVP) and
\emph{rank-adaptive two-site TDVP} (2TDVP), including applications to long-range
Hamiltonians~\cite{Haegeman_2016,Zaletel_2015,LiGleisVonDelft2024,Sander_TJM_2025}.  In this work, its success
in continuous dynamics motivates us to translate this projection picture to
circuit simulation by representing each discrete gate by a continuous generator
path and using the projected path itself to apply the gate.  This translation
raises the questions addressed here: how much of the MPS determines one gate's
projected action, whether the current MPS contains the required variational direction, and
how the continuous path can be realized as a practical update.  Variational
circuit constructions based on infinitesimal transformations predate this
work~\cite{PhysRevLett.100.130501}, but they do not provide this gate-specific
MPS characterization.

We establish an exact gate-window characterization of the instantaneous
projected dynamics, where both the fixed-rank and two-site projected actions of a
finite-support generator are determined within its contiguous window, while
the exterior MPS enters through the boundary information.  This recovers
instantaneous two-site exactness for adjacent gates~\cite{Haegeman_2016} and
gives a sufficient rank-one condition under which a long-range projected
action vanishes~\cite{Yang_2020}.  At smooth fixed rank, TDVP and ideal
nearest-point projection of the exact gate endpoint also agree to first order,
without implying equal endpoints of the gate path.

Additionally, we introduce a practical implementation of this structure as a rank-adaptive projected gate update.
We test the theoretical identities numerically.
Across several tested circuits, we see that the projected update incurs runtime overhead relative to
the established baselines and shows no consistent accuracy--cost
advantage.  The result is therefore a gate-specific framework and tested
update, rather than a general replacement for all MPS circuit simulation methods.
An open-source implementation of the gate-local update studied here is available as part of \emph{YAQS}~\cite{YAQS} within the \emph{Munich Quantum Toolkit}
(MQT)~\cite{wille_mqt2024}.

The remainder of the paper is organized as follows.
\cref{sec:continuous-tdvp} reviews the MPS and TDVP framework.
\cref{sec:gates-projected-evolution} introduces the projected gate path
and establishes the identities for localizing the projected update to a gate, then places nearest-neighbor
exactness and the long-range support obstruction in this framework.
\cref{sec:apply-compress-relation} compares the projected path with standard
apply-and-compress through an ideal fixed-rank endpoint projection of the gate's path.
\cref{sec:practical-gate-update}
constructs the practical gate-local update, and \cref{sec:results} tests
its analytical predictions and behavior. Finally,
\cref{sec:discussion} discusses their implications and limitations.

\section{Theoretical background}
\label{sec:continuous-tdvp}
This section reviews the MPS representation and the geometric ingredients of
projected evolution. We first introduce bond dimensions and the manifold of
states with fixed Schmidt ranks. Its tangent space defines fixed-rank TDVP,
while the two-site extension provides additional directions for rank-adaptive
evolution.

\subsection{Matrix product states}

A general state of \(N\) sites with local Hilbert space dimension \(d\) has
\(d^N\) coefficients in a product basis.  An MPS factors this coefficient
array into a chain of smaller tensors.  For open boundaries, the representation
is~\cite{Schollw_ck_2011}
\begin{equation}
    \ket{\Psi[A]}
    =
    \sum_{\sigma_1=1}^{d}\cdots\sum_{\sigma_N=1}^{d}
    A_1^{\sigma_1}A_2^{\sigma_2}\cdots A_N^{\sigma_N}
    \ket{\sigma_1,\ldots,\sigma_N},
    \label{eq:mps-definition}
\end{equation}
where \(\sigma_n\in\{1,\ldots,d\}\) labels a local basis state at site \(n\). For each value of \(\sigma_n\), the three-index site tensor contributes the matrix \(A_n^{\sigma_n}\in\mathbb{C}^{\chi_{n-1}\times\chi_n}\). Here, \(\chi_n\) is the dimension of the bond between sites \(n\) and \(n+1\); for consistency, the open ends have \(\chi_0=\chi_N=1\).
Tensor contraction along the neighboring auxiliary indices produces the scalar coefficient of
\(\ket{\sigma_1,\ldots,\sigma_N}\).  

This tensor representation is not unique.  On any internal bond \(n\), an
invertible change of bond basis \(G\) can be inserted into one tensor and
canceled in the next,
\(A_n^{\sigma_n}\mapsto A_n^{\sigma_n}G\) and
\(A_{n+1}^{\sigma_{n+1}}\mapsto G^{-1}A_{n+1}^{\sigma_{n+1}}\) such that the two
matrices cancel in Eq.~\eqref{eq:mps-definition}, leaving every physical
coefficient unchanged.  This internal MPS gauge freedom changes the tensor
coordinates but not the Hilbert-space vector represented by the
tensors~\cite{quant-ph/0608197,HaegemanGeometry2014}.

A convenient use of this freedom is to place the MPS in canonical form.
Relative to a chosen site or bond, the tensor blocks to its left and right act
as isometries.  Contracting either block with its conjugate gives the identity
on the open bond, or equivalently, the block states indexed by that bond are
orthonormal.  This makes the local norm matrix the identity and simplifies
local projections and effective evolution
problems~\cite{quant-ph/0608197,Schollw_ck_2011,Haegeman_2016}.

The physical meaning of \(\chi_n\) becomes clear by cutting the chain at this
bond, which we will refer to as cut $n$. This yields a bipartition
\([1:n]\mid[n+1:N]\), across which every 
state vector has a Schmidt decomposition
\begin{equation}
    \ket{\Psi}
    =
    \sum_{\alpha=1}^{\operatorname{rank}_n\Psi}
    \lambda_{\alpha}^{[n]}
    \ket{\Phi_{L,\alpha}^{[1:n]}}
    \ket{\Phi_{R,\alpha}^{[n+1:N]}},
    \qquad
    \lambda_{\alpha}^{[n]}>0,
    \label{eq:schmidt-decomposition}
\end{equation}
where the left and right Schmidt vectors are orthonormal states of the two
blocks.  The number of non-zero Schmidt coefficients is the Schmidt
rank  \(\operatorname{rank}_n\Psi\) across the cut.  It counts how many independent left--right mode pairs
are needed to represent the state there.

The MPS bond in Eq.~\eqref{eq:mps-definition} can carry at most \(\chi_n\)
such pairs, so \(\operatorname{rank}_n\Psi\leq\chi_n\).  Conversely, a state of Schmidt rank
\(\operatorname{rank}_n\Psi\) needs a bond dimension of at least \(\operatorname{rank}_n\Psi\). An MPS is therefore minimal on a
bond when these quantities agree.
Allowing larger bond dimensions enlarges the rank-bounded MPS family.  Every
state permitted by smaller bounds remains representable, and states with
higher ranks become available.  Rank-bounded families are therefore nested, whereas
sets of fixed exact rank are distinct. The manifold picture organizes
physical states by these intrinsic ranks, independently of gauge.

\subsection{Manifold picture and tangent space}

The bond dimensions restrict an MPS to a typically small part of Hilbert
space.  Fixing the actual Schmidt ranks selects a smooth MPS manifold, while
states with lower ranks lie on its boundary within the corresponding
rank-bounded family.  Let
\(\boldsymbol{\chi}=(\chi_1,\ldots,\chi_{N-1})\) be a vector of exact Schmidt
ranks that can occur on the chain, and set \(\chi_0=\chi_N=1\).
The non-zero states with exactly these ranks form a smooth embedded
manifold of not necessarily normalized states
\(\mathcal{M}_{\boldsymbol{\chi}}\) \cite{HaegemanGeometry2014, Holtz2012}
\begin{equation}
    \mathcal{M}_{\boldsymbol{\chi}}
    =
    \left\{
        \ket{\Phi}\neq 0:
        \operatorname{rank}_n\Phi=\chi_n
        \ \text{for all }n
    \right\},
    \label{eq:exact-rank-manifold}
\end{equation}

Given a state vector \(\ket{\Psi}\in\mathcal{M}_{\boldsymbol{\chi}}\), its Schmidt decomposition
across cut $n$ in \cref{eq:schmidt-decomposition} contains exactly \(\chi_n\)
non-zero terms. The left and right Schmidt vectors span the
corresponding Schmidt-support spaces.

Given a cut \(n\), let
\(\{\ket{\Phi_{L,\alpha}^{[1:n]}}\}_{\alpha=1}^{\chi_n}\) and
\(\{\ket{\Phi_{R,\alpha}^{[n+1:N]}}\}_{\alpha=1}^{\chi_n}\)
be orthonormal bases for the left and right Schmidt-support spaces of
\(\ket{\Psi}\) at that cut.  
Projecting the full block Hilbert spaces onto
these subspaces requires the projectors
\begin{subequations}
\label{eq:mps-support-projectors}
\begin{align}
    P_L^{[1:n]}
    &=
    \sum_{\alpha=1}^{\chi_n}
    \ket{\Phi_{L,\alpha}^{[1:n]}}
    \bra{\Phi_{L,\alpha}^{[1:n]}} \text{ and }
    \label{eq:mps-left-projector}
    \\
    P_R^{[n+1:N]}
    &=
    \sum_{\alpha=1}^{\chi_n}
    \ket{\Phi_{R,\alpha}^{[n+1:N]}}
    \bra{\Phi_{R,\alpha}^{[n+1:N]}}.
    \label{eq:mps-right-projector}
\end{align}
\end{subequations}
An empty boundary block has a one-dimensional Hilbert space, so its projector
is the scalar identity,
\(P_L^{[1:0]}=P_R^{[N+1:N]}=1\). Applied to an arbitrary block state, they retain its
component in the corresponding Schmidt-support space and remove the orthogonal
component.  Their definition depends only on the Schmidt-support spaces and
not on the particular bases chosen for them.  The rank vector
\(\boldsymbol{\chi}\) fixes
the dimensions of these spaces, but the spaces themselves depend on
\(\ket{\Psi}\).

This state dependence is central to the local geometry of the manifold.  At
\(\ket{\Psi}\), the tangent space
\(T_\Psi\mathcal{M}_{\boldsymbol{\chi}}\) consists of the derivatives of
smooth fixed-rank MPS curves through that state. It is the linear space of
first-order changes available without altering the exact ranks.

These changes can be represented by single-site tensor variations, whose span and orthogonal projector are denoted by
\begin{equation}
    \mathcal{V}_{\Psi}^{[1]}
    =
    T_\Psi\mathcal{M}_{\boldsymbol{\chi}},
    \qquad
    P_{\Psi}^{[1]}
    =
    P_{T_\Psi\mathcal{M}_{\boldsymbol{\chi}}},
    \label{eq:single-site-variation-space}
\end{equation}
where \(\mathcal{V}_{\Psi}^{[1]}\) is the variational space at \(\ket{\Psi}\). The first equality identifies this variation
space with the tangent space \(T_\Psi\mathcal{M}_{\boldsymbol{\chi}}\) of the
fixed-rank manifold at \(\ket{\Psi}\), and the operator \(P_{\Psi}^{[1]}\) is the
orthogonal projector onto this space.  The notation
\(P_{T_\Psi\mathcal{M}_{\boldsymbol{\chi}}}\) means the orthogonal projector
whose range is the tangent space \(T_\Psi\mathcal{M}_{\boldsymbol{\chi}}\).  The superscript
\([1]\) labels the single-site construction, while bracketed intervals such as
\([1:n]\) continue to label physical blocks.

A single-site tangent variation changes the tensor at its center while the
surrounding blocks remain within their current Schmidt-support spaces.  The
block projectors enforce this restriction, and an identity operator leaves the
central site unrestricted. Together, these operators give an explicit MPS
representation of \(P_{\Psi}^{[1]}\).

Let \(I_n\) denote the identity on the
\(d\)-dimensional physical Hilbert space of site \(n\).  The operator \(S_n\)
projects onto the space of variations centered at site \(n\), while
\(B_n\) projects onto the bond-centered overlap 
\begin{subequations}
\label{eq:tangent-projector-terms}
\begin{align}
    S_n
    &=
    P_L^{[1:n-1]}
    \otimes I_n
    \otimes P_R^{[n+1:N]},
    \label{eq:tangent-site-term}
    \\
    B_n
    &=
    P_L^{[1:n]}
    \otimes P_R^{[n+1:N]}
    \label{eq:tangent-bond-term}
\end{align}
\end{subequations}
shared by the neighboring
single-site spaces. Subtracting the bond-centered projectors once gives the established
orthogonal projector onto the fixed-rank tangent space~\cite{Haegeman_2016} gives
\begin{equation}
    P_{\Psi}^{[1]}
    =
    \sum_{n=1}^{N}S_n
    -
    \sum_{n=1}^{N-1}B_n.
    \label{eq:tangent-projector-compact}
\end{equation}
The tangent space and its orthogonal
projector therefore provide the state-dependent linear geometry needed to
define evolution along this manifold at fixed rank.

\subsection{Fixed-rank continuous time evolution}

Under a Hermitian full-chain generator \(H\), exact evolution has velocity in Hilbert space
\(-\mathrm{i}H\ket{\Psi}\).  Fixed-rank TDVP replaces this velocity by its
orthogonal projection onto the current tangent space, giving the projected Schrödinger equation
\begin{equation}
    \partial_t\ket{\Psi(t)}
    =
    -\mathrm{i}\,
    P_{\Psi(t)}^{[1]}
    H\ket{\Psi(t)}.
    \label{eq:tdvp-equation}
\end{equation}
Here \(t\) is the continuous evolution parameter (reflecting physical time in this setting), and the projector is rebuilt
from \(\ket{\Psi(t)}\) at each step as the trajectory advances.

Equivalently, let \(\eta\) be any candidate tangent vector and let
\(\lVert\cdot\rVert\) denote the Hilbert space norm.  The TDVP direction is the
unique tangent vector that minimizes its distance from the exact direction,
\begin{equation}
    -\mathrm{i}
    P_{\Psi(t)}^{[1]}
    H\ket{\Psi(t)}
    =
    \underset{\eta\in
        \mathcal{V}_{\Psi(t)}^{[1]}}{\argmin}
    \left\lVert
        \eta+\mathrm{i}H\ket{\Psi(t)}
    \right\rVert .
    \label{eq:tdvp-instantaneous-minimization}
\end{equation}
This is the precise sense in which fixed-rank TDVP is instantaneously
optimal: it minimizes the residual at the current state.
As the state moves along the curved MPS manifold, its tangent space changes, so each
instantaneous projection affects those that follow.  Integrating these
instantaneously minimizing directions is therefore not guaranteed to end at
the MPS closest to the exact state after a finite evolution.  This distinction
becomes important when the evolution represents a finite evolution rather
than an infinitesimal step.
Additionally, the relation \(P_\Psi^{[1]}\ket{\Psi}=\ket{\Psi}\), together with
self-adjointness of the projector and Hermiticity of \(H\), gives
\(\frac{\diff}{\diff t}\lVert\Psi\rVert^2
=-2\operatorname{Re}(\mathrm{i}\langle\Psi|H|\Psi\rangle)=0\), so the flow is norm-preserving as long as it remains on the current manifold
\(\mathcal{M}_{\boldsymbol{\chi}}\).

\subsection{Rank-adaptive continuous time evolution}

Fixed-rank TDVP can update the Schmidt vectors, but it cannot change their
number across any bond. The two-site variation space used by 2TDVP enlarges
the available directions by treating neighboring sites \(n,n+1\) as one
active block.  The degrees of freedom within that pair are unrestricted, allowing its internal Schmidt rank to adapt, while
the exterior blocks remain confined to their current Schmidt-support spaces. The corresponding pair projector, which generalizes \cref{eq:tangent-projector-terms}, is
\begin{equation}
    \begin{aligned}
        K_{n:n+1}
        &=
        P_L^{[1:n-1]}
        \otimes I_n
        \otimes I_{n+1}
        \\
        &\quad\otimes
        P_R^{[n+2:N]},
        \qquad
        1\leq n\leq N-1.
    \end{aligned}
    \label{eq:two-site-center-term}
\end{equation}
Here, the subscript \(n:n+1\) labels the contiguous two-site block.  Since every
factor in \cref{eq:two-site-center-term} is an orthogonal projector or an
identity, \(K_{n:n+1}\) is itself an orthogonal projector.  Its range is the
pair-centered subspace of full-chain variations compatible with the current
exterior environments.  Combining these subspaces over all neighboring pairs
defines
\begin{equation}
    \mathcal{V}_{\Psi}^{[2]}
    =
    \sum_{n=1}^{N-1}
    \operatorname{range}\!\left(K_{n:n+1}\right).
    \label{eq:two-site-variation-space}
\end{equation}
The sum of subspaces denotes all linear combinations of vectors drawn from the
individual ranges centered on each contiguous pair.  The superscript \([2]\) denotes this
two-site construction, in parallel with the single-site space
\(\mathcal{V}_{\Psi}^{[1]}\).

Importantly, neighboring pair spaces are not independent since the blocks \(n-1:n\) and
\(n:n+1\) share the single-site variations centered at site \(n\).  For
\(2\leq n\leq N-1\), the nesting of the Schmidt-support spaces makes their
projectors commute and gives
\begin{equation}
    K_{n-1:n}K_{n:n+1}
    =
    K_{n:n+1}K_{n-1:n}
    =:
    S_n.
    \label{eq:two-site-overlap}
\end{equation}
The product of commuting orthogonal projectors projects onto the intersection
of their ranges.  Their common range is therefore exactly the
single-site-centered space selected by \(S_n\).  This leads to the orthogonal projector onto
\(\mathcal{V}_{\Psi}^{[2]}\), denoted by \(P_{\Psi}^{[2]}\), as the
established alternating decomposition~\cite{Haegeman_2016}
\begin{equation}
    P_{\Psi}^{[2]}
    =
    \sum_{n=1}^{N-1}K_{n:n+1}
    -
    \sum_{n=2}^{N-1}S_n.
    \label{eq:two-site-variation-projector}
\end{equation}
Here, the single-site index on \(S_n\) thus identifies the shared center of the
neighboring pair spaces rather than another two-site contribution.

At this stage, \(\mathcal{V}_{\Psi}^{[2]}\) is simply a larger space of
candidate directions for the current MPS: the projection specifies a direction
rather than a new MPS representation.  After the pair is evolved, its two-site
tensor may have a higher Schmidt rank across the internal cut.  The SVD exposes
that rank, and retaining the additional non-zero components
enlarges the stored bond and its Schmidt support.  The pair projection and
this subsequent factorization together form the basis of rank-adaptive
two-site TDVP.

The two projectors therefore play complementary roles.
\(P_{\Psi}^{[1]}\) describes motion on one fixed-rank manifold, while
\(P_{\Psi}^{[2]}\) supplies the additional pair-centered directions used by an
adaptive update.  These established constructions provide the geometric
background for the gate-level analysis and numerical methods developed
below.  We now turn to applying them to evolutions generated by
gates and ask how this manifold picture plays out in practice.

\section{Gate-local projected evolution}
\label{sec:gates-projected-evolution}

To answer this question, we first represent a finite two-site gate by a
continuous evolution defined by its generator and define its fixed-rank and
two-site projected directions.  We then prove exact identities for their
instantaneous actions, that is, their infinitesimal behavior.  In particular,
we show that the projected evolution localizes to a window around the gate.
Within this gate setting, the section also recovers nearest-neighbor
instantaneous exactness and a long-range obstruction previously observed in
many-body simulation.

\subsection{Quantum gates as continuous evolution}

A quantum gate specifies a discrete unitary
operation rather than a unique continuous path from the identity to that
unitary.  Every finite-dimensional unitary, however, admits a Hermitian
generator that defines such an evolution.  For a gate \(U_g\) acting on sites
\(1\leq a<b\leq N\), we choose a generator \(H_g\) and write
\begin{equation}
    \begin{aligned}
        U_g
        &=
        \exp\!\left(-\mathrm{i}H_g\right),
        \\
        U_g(s)
        &=
        \exp\!\left(-\mathrm{i}sH_g\right),
        \qquad s\in[0,1].
    \end{aligned}
    \label{eq:gate-path}
\end{equation}
The dimensionless coordinate \(s\) measures progress through the gate rather
than physical time, with \(U_g(0)=I\) and \(U_g(1)=U_g\).  Without projection,
the state \(\ket{\Psi_{\mathrm{ex}}(s)}=U_g(s)\ket{\Psi}\) follows the exact
velocity \(-\mathrm{i}H_g\ket{\Psi_{\mathrm{ex}}(s)}\) and reaches
\(U_g\ket{\Psi}\) at \(s=1\).  For
\(\ket{\Psi}\in\mathcal{M}_{\boldsymbol{\chi}}\), fixed-rank projected
evolution replaces that velocity at every point by its component in the
current MPS tangent space,
\begin{equation}
    \begin{aligned}
        \partial_s\ket{\Psi_g(s)}
        &=
        -\mathrm{i}\,
        P_{\Psi_g(s)}^{[1]}
        H_g\ket{\Psi_g(s)},
        \\
        \ket{\Psi_g(0)}
        &=
        \ket{\Psi}.
    \end{aligned}
    \label{eq:projected-gate-flow}
\end{equation}
Here, the projector in \cref{eq:projected-gate-flow} is evaluated from the
current MPS as the path advances.  The resulting equation is a smooth \emph{ordinary
differential equation} (ODE) on \(\mathcal{M}_{\boldsymbol{\chi}}\) wherever its
solution exists.

At a given MPS \(\ket{\Phi}\) describing the current state, the corresponding
two-site projected velocity is \(-\mathrm{i}P_{\Phi}^{[2]}H_g\ket{\Phi}\).  The
two constructions play complementary roles: the fixed-rank equation defines a
smooth trajectory and supports trajectory-level statements and the comparison
with direct projection, while the two-site expression specifies the larger
instantaneous direction used to construct the rank-adaptive update.  The
two-site expression is not a smooth ODE through different rank strata since,
after a finite pair update, the two-site tensor may have a higher intrinsic
Schmidt rank, causing a jump to a higher-rank manifold.  The practical
algorithm changes the stored bond dimension when that tensor is factorized and
the additional singular directions are retained.

For the analysis below, it is useful to make the two-site operator structure of
the chosen generator explicit.  A general generator acting on the two gate
sites has the operator Schmidt decomposition
\begin{equation}
    \begin{aligned}
        H_g
        &=
        \sum_{\mu=1}^{D_g}
        I^{\otimes(a-1)}
        \otimes A_\mu^{(a)}
        \\
        &\quad\otimes I^{\otimes(b-a-1)}
        \otimes B_\mu^{(b)}
        \otimes I^{\otimes(N-b)}.
    \end{aligned}
    \label{eq:generic-two-site-schmidt-decomposition}
\end{equation}
Here, \(I\) is the single-site identity, \(I^{\otimes0}=1\), and
\(D_g\leq d^2\), with \(D_g\leq4\) for qubits.  The analysis below allows any
\(D_g\).  An important special case is a product generator,
\begin{equation}
    H_g
    =
    I^{\otimes(a-1)}
    \otimes G_a^{(g)}
    \otimes I^{\otimes(b-a-1)}
    \otimes G_b^{(g)}
    \otimes I^{\otimes(N-b)}
    \label{eq:product-gate-generator}
\end{equation}
For \(H_g\ne0\), this is the \(D_g=1\) case.
Viewed along the MPS ordering, \cref{eq:generic-two-site-schmidt-decomposition}
gives an exact \emph{matrix product operator} (MPO) representation with bond dimension \(D_g\) across the gate
span and bond dimension one outside it.  Thus the product generators used below
have exact bond-one MPO representations. Note that this statement concerns \(H_g\), not
generally the finite gate \(U_g\).

Pauli rotations provide a family of gates with product generators:
\begin{equation}
    \begin{aligned}
        R_{\alpha\beta}(\theta)
        &=
        \exp\!\left[
            -\mathrm{i}\frac{\theta}{2}
            \sigma_\alpha\otimes\sigma_\beta
        \right],
        \\
        h_{\alpha\beta}(\theta)
        &=
        \frac{\theta}{2}
        \sigma_\alpha\otimes\sigma_\beta,
        \qquad
        \alpha,\beta\in\{x,y,z\}.
    \end{aligned}
    \label{eq:rotation-gate-generators}
\end{equation}
Here, \(\sigma_x=X\), \(\sigma_y=Y\), and \(\sigma_z=Z\) are the Pauli
matrices.  A \emph{controlled-NOT} (CNOT) gate with the first of the two gate sites as control likewise
admits the product generator
\begin{equation}
    h_{\mathrm{CNOT}}
    =
    \frac{\pi}{4}
    (I-Z)\otimes(I-X),
    \qquad
    \mathrm{CNOT}
    =
    e^{-\mathrm{i}h_{\mathrm{CNOT}}}.
    \label{eq:cnot-gate-generator}
\end{equation}
Embedding these local generators at sites \(a\) and \(b\), with identities on
the remaining sites, gives the full-chain \(H_g\).  Thus the same construction
also assigns a continuous path to a gate without an external parameter such as the angle in the rotation gates.

The generator need not be unique because the eigenphases of \(U_g\) are
defined modulo \(2\pi\). Different choices can trace different intermediate paths
and therefore produce different projected trajectories. The generator path is
thus part of the approximation.  Additional constructions of
generators for single- and multiqubit logic gates, including controlled gates
with multiple control and target qubits, are given in
Refs.~\cite{muradian2005generatorsrootsquantumlogic,milbradt2026shortnotegeneratorscontrolled}.
With the path fixed, we now identify which terms of the apparently full-chain
projector can affect its instantaneous velocity.

\subsection{Gate-window locality}
\label{sec:gate-window-locality}

Although \(H_g\) acts only within the interval between the gate sites, both the
fixed-rank and two-site projectors are written as sums over the full chain.  We
now prove that their complete instantaneous actions are reproduced by terms
within the same gate window.  Individual terms can extend across its
boundaries, but the exterior contributions cancel in the alternating sums.

Let \(N\geq2\) and \(1\leq a<b\leq N\).  The gate window is the contiguous
interval
\begin{equation}
    W_g=[a:b]
    =
    \{a,a+1,\ldots,b\}.
    \label{eq:gate-window}
\end{equation}
We assume that \(H_g\) acts trivially outside this window, so that
\begin{equation}
    H_g
    =
    I^{\otimes(a-1)}
    \otimes h_g^{[a:b]}
    \otimes I^{\otimes(N-b)}
    \label{eq:gate-window-supported-generator}
\end{equation}
for an operator \(h_g^{[a:b]}\) on the window.  For the two-site generators
considered here, Eq.~\eqref{eq:generic-two-site-schmidt-decomposition} gives
this form with nonidentity factors only at the endpoint sites \(a\) and \(b\).

The two restricted sums needed below are
given by
\begin{align}
    \mathcal{P}^{[1]}_{\Psi,W_g}
    &=
    \sum_{n=a}^{b}S_n
    -
    \sum_{n=a}^{b-1}B_n,
    \label{eq:gate-window-tangent-operator}
    \\
    \mathcal{P}^{[2]}_{\Psi,W_g}
    &=
    \sum_{n=a}^{b-1}K_{n:n+1}
    -
    \sum_{n=a+1}^{b-1}S_n.
    \label{eq:two-site-window-operator}
\end{align}
These are restrictions of the standard projector decompositions, not
independent projector formulas on arbitrary Hilbert-space vectors.  Their
Schmidt-support projectors are constructed from the full MPS, and empty sums
are understood to vanish.

\begin{theorem}[Gate-window locality of projected generator actions]
\label{thm:gate-window-locality}
Let \(N\geq2\), \(1\leq a<b\leq N\), and let
\(\ket{\Psi}\in\mathcal{M}_{\boldsymbol{\chi}}\) be a smooth fixed-rank MPS,
with \(H_g\) acting trivially outside \(W_g=[a,b]\).  Then
\begin{align}
    P_{\Psi}^{[1]}
    H_g\ket{\Psi}
    &=
    \mathcal{P}^{[1]}_{\Psi,W_g}
    H_g\ket{\Psi},
    \label{eq:gate-window-vector-field-equality}
    \\
    P_{\Psi}^{[2]}
    H_g\ket{\Psi}
    &=
    \mathcal{P}^{[2]}_{\Psi,W_g}
    H_g\ket{\Psi}.
    \label{eq:two-site-window-localization}
\end{align}
Thus the net fixed-rank and two-site projected actions both localize to the
same physical window \(W_g\).
\end{theorem}

\begin{proof}
Set \(\ket{X_g}=H_g\ket{\Psi}\).  For \(n<a\), define, with identities on all
omitted sites,
\begin{equation}
    Q_n
    =
    P_L^{[1:n-1]}\otimes I_n,
    \qquad
    \widetilde Q_n
    =
    P_L^{[1:n]}.
    \label{eq:left-support-factors}
\end{equation}
Each Schmidt-support projector, tensored with the identity on the complementary
block, fixes \(\ket{\Psi}\).  Both operators in
Eq.~\eqref{eq:left-support-factors} therefore fix \(\ket{\Psi}\), and they
commute with \(H_g\) because their supports lie strictly to its left.  Hence
\(Q_n\ket{X_g}=\ket{X_g}=\widetilde Q_n\ket{X_g}\).
Multiplying by the common right Schmidt-support projector gives the first
identity below.  The analogous argument with exterior right-block factors
gives the second
expression
\begin{align}
    S_n\ket{X_g}
    &=
    B_n\ket{X_g},
    \qquad n<a ,
    \label{eq:left-exterior-cancellation}
    \\
    S_n\ket{X_g}
    &=
    B_{n-1}\ket{X_g},
    \qquad n>b.
    \label{eq:right-exterior-cancellation}
\end{align}
Consequently,
\begin{equation}
    \begin{aligned}
        \left(P_{\Psi}^{[1]}-\mathcal{P}^{[1]}_{\Psi,W_g}\right)\ket{X_g}
        &=
        \sum_{n=1}^{a-1}(S_n-B_n)\ket{X_g}
        \\
        &\quad+
        \sum_{n=b+1}^{N}(S_n-B_{n-1})\ket{X_g}
        =0.
    \end{aligned}
    \label{eq:fixed-rank-exterior-regrouping}
\end{equation}

For the two-site sum, the same exterior-support argument 
provides
\begin{align}
    K_{n:n+1}\ket{X_g}
    &=
    S_{n+1}\ket{X_g},
    \qquad n<a,
    \label{eq:two-site-left-cancellation}
    \\
    K_{n:n+1}\ket{X_g}
    &=
    S_n\ket{X_g},
    \qquad n\geq b.
    \label{eq:two-site-right-cancellation}
\end{align}
The left identity includes the boundary-straddling pair
\(K_{a-1:a}\ket{X_g}=S_a\ket{X_g}\) when \(a>1\), and the right identity
includes \(K_{b:b+1}\ket{X_g}=S_b\ket{X_g}\) when \(b<N\).  These terms need
not vanish individually.  Regrouping the full two-site decomposition gives
\begin{equation}
    \begin{aligned}
        \left(P_{\Psi}^{[2]}-\mathcal{P}^{[2]}_{\Psi,W_g}\right)\ket{X_g}
        &=
        \sum_{n=1}^{a-1}(K_{n:n+1}-S_{n+1})\ket{X_g}
        \\
        &\quad+
        \sum_{n=b}^{N-1}(K_{n:n+1}-S_n)\ket{X_g}
        =0.
    \end{aligned}
    \label{eq:two-site-exterior-regrouping}
\end{equation}
The left cancellation sum is empty when \(a=1\), and the right cancellation
sum is empty when \(b=N\).
\end{proof}

Both restricted actions remain embedded in the effective space
\begin{equation}
    \begin{aligned}
        \mathcal{H}^{\mathrm{eff}}_{W_g}(\Psi)
        &=
        \operatorname{range}\!\left(P_L^{[1:a-1]}\right)
        \otimes
        \left(
            \bigotimes_{n=a}^{b}\mathbb{C}^{d}_n
        \right)
        \\
        &\quad\otimes
        \operatorname{range}\!\left(P_R^{[b+1:N]}\right),
    \end{aligned}
    \label{eq:gate-window-effective-space}
\end{equation}
whose boundary spaces are the Schmidt-support spaces across the two cuts that
bound \(W_g\). Localization removes exterior terms from the summed actions
while retaining the exterior Schmidt spaces inherited from the full MPS.

\subsection{Exactness and obstruction}
Theorem~\ref{thm:gate-window-locality} tells us where to carry out the update, but not how well
the projected update still follows the gate path.  For neighboring sites, we show that
2TDVP contains the gate's full instantaneous direction, as established in
Ref.~\cite{Haegeman_2016}.  For long-range gates, this exactness depends on the current state and need not hold.

\begin{corollary}[Nearest-neighbor exactness]
\label{cor:nearest-neighbour-two-site-exactness}
Under the assumptions of \cref{thm:gate-window-locality}, if
\(b=a+1\), then
\begin{equation}
    H_g\ket{\Psi}
    \in
    \operatorname{range}\!\left(K_{a:a+1}\right)
    \subseteq
    \mathcal{V}_{\Psi}^{[2]},
    \label{eq:nearest-neighbour-two-site-inclusion}
\end{equation}
and consequently
\begin{equation}
    P_{\Psi}^{[2]}H_g\ket{\Psi}
    =
    H_g\ket{\Psi}.
    \label{eq:nearest-neighbour-two-site-exactness}
\end{equation}
\end{corollary}
\begin{proof}
The exterior Schmidt-support projectors in \(K_{a:a+1}\) act as the identity on
\(\ket{\Psi}\) and commute with a generator that acts nontrivially only on the
physical sites \(a,a+1\).
Hence \(K_{a:a+1}H_g\ket{\Psi}=H_g\ket{\Psi}\), so
\(H_g\ket{\Psi}\in
\operatorname{range}\!\left(K_{a:a+1}\right)\).  The two-site variation
space is the sum of the ranges of the \(K_{n:n+1}\), and
\(P_{\Psi}^{[2]}\) is its orthogonal projector.  Therefore
\(\operatorname{range}\!\left(K_{a:a+1}\right)\subseteq
\mathcal{V}_{\Psi}^{[2]}\), and the projector acts as the identity on
\(H_g\ket{\Psi}\).
\end{proof}
Thus, the instantaneous two-site projection residual vanishes for a
nearest-neighbor generator, including at the boundary of an MPS and for \(N=2\).

There is, however, no corresponding guarantee for a long-range gate.
For a general state, the two-site projection residual diagnoses whether
the exact generator direction is available. Along the exact trajectory
\(\ket{\Psi_{\mathrm{ex}}(s)}=e^{-\mathrm{i}sH_g}\ket{\Psi}\), the residual
vanishes at a given \(s\in[0,1]\) if and only if
\begin{equation}
    \left(
        I-P_{\Psi_{\mathrm{ex}}(s)}^{[2]}
    \right)
    H_g\ket{\Psi_{\mathrm{ex}}(s)}
    =
    0.
    \label{eq:two-site-exact-recovery-condition}
\end{equation}
Here \(P_{\Psi_{\mathrm{ex}}(s)}^{[2]}\) is constructed from the exact state's
instantaneous Schmidt-support projectors, which may change along the path.

The freezing of product states due to missing two-site variational directions was
previously demonstrated for disconnected long-range interactions
~\cite{Yang_2020}.  While a state with a rank-one bond does not inherently make every update
fail, we give the following gate-level condition for this obstruction.

\begin{proposition}[A rank-one obstruction for long-range gates]
\label{prop:rank-one-long-range-gate-projection}
Let \(N\geq3\) and \(1\leq a<b\leq N\), with \(b>a+1\), and let
\(\ket{\Psi}\) be a normalized MPS whose exact Schmidt ranks satisfy
\(\chi_a=\chi_{b-1}=1\).  The state vector then admits a factorization into normalized
blocks
\begin{equation}
    \ket{\Psi}
    =
    \ket{\Psi_L}_{1:a}
    \otimes
    \ket{\Psi_M}_{a+1:b-1}
    \otimes
    \ket{\Psi_R}_{b:N}.
    \label{eq:rank-one-three-block-factorization}
\end{equation}
Define
\(\Pi_L=\ket{\Psi_L}\!\bra{\Psi_L}\),
\(\Pi_R=\ket{\Psi_R}\!\bra{\Psi_R}\), and
\(\Pi_L^\perp=I_L-\Pi_L\),
\(\Pi_R^\perp=I_R-\Pi_R\), where \(I_L\), \(I_M\), and \(I_R\) are the
identities on the three blocks.  We denote the projector onto their jointly
orthogonal sector by
\begin{equation}
    \mathcal Q_{LR}
    =
    \Pi_L^\perp\otimes I_M\otimes\Pi_R^\perp.
    \label{eq:joint-outer-block-projector}
\end{equation}
If \(H_g\) acts only on sites \(a\) and \(b\),
then, for \(\ket{X_g}=H_g\ket{\Psi}\),
\begin{equation}
    \begin{aligned}
        P_\Psi^{[1]}\ket{X_g}
        =
        P_\Psi^{[2]}\ket{X_g}
        &=
        \left(I-\mathcal Q_{LR}\right)\ket{X_g},
        \\
        \left(I-P_\Psi^{[1]}\right)\ket{X_g}
        =
        \left(I-P_\Psi^{[2]}\right)\ket{X_g}
        &=
        \mathcal Q_{LR}\ket{X_g}.
    \end{aligned}
    \label{eq:rank-one-long-range-gate-projection}
\end{equation}
\end{proposition}

\begin{proof}
The two rank-one cuts give the three-block factorization in
Eq.~\eqref{eq:rank-one-three-block-factorization}.  Decompose \(\ket{X_g}\)
using \(I_L=\Pi_L+\Pi_L^\perp\) and
\(I_R=\Pi_R+\Pi_R^\perp\).  The component that leaves both outer blocks
unchanged is proportional to \(\ket{\Psi}\).  Projecting either outer block
back onto its occupied state contracts \(H_g\) with that block and leaves an
operator acting only on the opposite endpoint site.  The two components that
change only one outer block therefore lie in the single-site spaces selected by
\(S_a\) and \(S_b\).  These first three components all belong to
\(\mathcal V_\Psi^{[1]}\subseteq\mathcal V_\Psi^{[2]}\).  The remaining component
\(\mathcal Q_{LR}\ket{X_g}\) lies in
\(\operatorname{range}(\Pi_L^\perp)
\otimes\operatorname{span}\!\{\ket{\Psi_M}\}\otimes
\operatorname{range}(\Pi_R^\perp)\).  For \(n\leq b-2\), the pair
\(n:n+1\) does not contain site \(b\), and the rank-one cut at \(b-1\) fixes
the right-block factor of every vector in
\(\operatorname{range}(K_{n:n+1})\) to \(\ket{\Psi_R}\).  The component above
is orthogonal to that factor.  Similarly, for \(n\geq a+1\), the rank-one cut
at \(a\) fixes the left-block factor to \(\ket{\Psi_L}\), which is orthogonal
to \(\operatorname{range}(\Pi_L^\perp)\).  These two ranges cover every
adjacent pair because \(b>a+1\), so the component is orthogonal to
\(\mathcal V_\Psi^{[2]}\), and thus also to
\(\mathcal V_\Psi^{[1]}\).  Orthogonality of the four components gives
Eq.~\eqref{eq:rank-one-long-range-gate-projection}.
\end{proof}

\begin{corollary}[Freezing under the rank-one obstruction]
\label{cor:rank-one-long-range-gate-freezing}
Under the assumptions of
Proposition~\ref{prop:rank-one-long-range-gate-projection}, with
\(\ket{X_g}\neq0\), both projected actions vanish if and only if
\begin{equation}
    \ket{X_g}
    =
    \mathcal Q_{LR}\ket{X_g}.
    \label{eq:rank-one-obstruction-condition}
\end{equation}
When this condition holds, the fixed-rank projected gate evolution initialized
at \(\ket{\Psi}\) is stationary.
\end{corollary}

\begin{proof}
The equivalence follows directly from
Eq.~\eqref{eq:rank-one-long-range-gate-projection}.  Under this condition the
vector field in Eq.~\eqref{eq:projected-gate-flow} vanishes at \(\ket{\Psi}\),
and it is smooth on the exact-rank manifold.  Uniqueness of the initial-value
problem therefore makes the constant curve the projected evolution initialized
at \(\ket{\Psi}\).
\end{proof}

For a product generator \(H_g=A^{(a)}B^{(b)}\), the condition in
Eq.~\eqref{eq:rank-one-obstruction-condition} holds whenever the non-zero
vectors \(A^{(a)}\ket{\Psi_L}\) and \(B^{(b)}\ket{\Psi_R}\) are orthogonal to
\(\ket{\Psi_L}\) and \(\ket{\Psi_R}\), respectively. 

For example, take the four-qubit instance where \(\ket{\Psi}=\ket{0000}\) and
\(H_g=\sigma_x^{(1)}\sigma_x^{(4)}\).  It satisfies the condition because
\begin{equation}
    \begin{aligned}
        \ket{X_g}&=
        \ket{1,0,0,1},
        &
        \lVert\ket{X_g}\rVert&=1,
        \\
        P_\Psi^{[1]}\ket{X_g}&=0,
        &
        P_\Psi^{[2]}\ket{X_g}&=0.
    \end{aligned}
    \label{eq:four-qubit-zero-action-example}
\end{equation}

Under the condition in
Corollary~\ref{cor:rank-one-long-range-gate-freezing}, the vanishing projected
action is not a cancellation between non-zero projector terms.  Since
\(P_\Psi^{[2]}\) is the orthogonal projector onto
\(\mathcal V_\Psi^{[2]}\), the identity
\(P_\Psi^{[2]}\ket{X_g}=0\) implies
\(\ket{X_g}\perp\mathcal V_\Psi^{[2]}\), and hence
\begin{equation}
    \begin{aligned}
        K_{n:n+1}\ket{X_g}
        &=
        0,
        &&1\leq n\leq N-1,
        \\
        S_n\ket{X_g}
        &=
        0,
        &&2\leq n\leq N-1.
    \end{aligned}
    \label{eq:two-site-zero-subflow-actions}
\end{equation}
Here the single-site space selected by \(S_n\) is a subspace of
\(\mathcal V_\Psi^{[2]}\).
Thus every term in the restricted instantaneous decomposition vanishes at such
a state. 

The obstruction is determined by the represented Schmidt-support spaces.
For this class of zero-action states, increasing \(\chi_{\max}\) while retaining only
the current non-zero Schmidt supports cannot create the missing direction.
Established Krylov-enrichment and controlled bond-expansion methods instead
enlarge the available space~\cite{Yang_2020,LiGleisVonDelft2024}.  Adopting
such enrichment here, using a nonminimal initialization, or deliberately
retaining zero-singular-value basis directions would define a different
algorithm and is therefore left as a potential direction for future work.

\subsection{Relation to established methods}
\label{sec:apply-compress-relation}

Projected gate evolution constrains the MPS throughout the generator path, while several established methods (which we call \emph{apply-and-compress})
instead work with discrete gate operations and reduce
the enlarged tensors they produce.  The latter description covers several
numerical procedures which, rather than stepping along this path, approximate the final endpoint.  TEBD truncates an enlarged active bond with an
SVD~\cite{Paeckel_2019}, MPO contract-and-truncate methods reduce enlarged
bonds during a sweep~\cite{Stoudenmire_2010,Sander_EC_2025, tensornetworkorg, CamanoEpperlyTropp2025, Milbradt2026}, and
variational compression fits a rank-bounded MPS to the gate-applied target
through repeated local optimizations~\cite{Stoudenmire_2010,Verstraete_2004}.

To obtain an algorithm-independent comparison, we use the locally nearest
fixed-rank MPS as an ideal endpoint reference.  Starting from
\(\ket{\Psi}\in\mathcal M_{\boldsymbol{\chi}}\) and a sufficiently small gate
strength \(\epsilon\), this reference maps the exact endpoint
\(e^{-\mathrm{i}\epsilon H_g}\ket{\Psi}\) to the nearest state in a small
neighborhood of \(\ket{\Psi}\) on the same exact-rank manifold. This asks
whether ideal endpoint reduction and continuous tangent projection have the
same initial, infinitesimal action on the quantum state.
The following lemma makes this local nearest-point map well-defined and
identifies its first-order action with the TDVP projector.  All state-vector
remainder estimates below are in the Hilbert norm.

\begin{lemma}[Local metric projection]
\label{lem:local-metric-projection}
Let \(\ket{\Psi}\in\mathcal{M}_{\boldsymbol{\chi}}\).
There exists an ambient neighborhood \(\mathcal{U}\) of \(\ket{\Psi}\)
on which the Hilbert-norm nearest-point map
\(\Pi_{\mathcal{M}_{\boldsymbol{\chi}}}\) onto a small manifold neighborhood
of \(\ket{\Psi}\) is well-defined and \(C^2\).  For every ambient vector
\(\ket{v}\) and all sufficiently small real \(\epsilon\),
\begin{equation}
    \begin{aligned}
        \Pi_{\mathcal{M}_{\boldsymbol{\chi}}}(\ket{\Psi})
        &=
        \ket{\Psi},
        \\
        \Pi_{\mathcal{M}_{\boldsymbol{\chi}}}
        \bigl(\ket{\Psi}+\epsilon\ket{v}\bigr)
        &=
        \ket{\Psi}
        +\epsilon\,P_{\Psi}^{[1]}\ket{v}
        +\mathcal{O}(\epsilon^2).
    \end{aligned}
    \label{eq:local-metric-projection}
\end{equation}
\end{lemma}

\begin{proof}
The exact-rank set \(\mathcal{M}_{\boldsymbol{\chi}}\) is a smooth embedded
complex submanifold of Hilbert space near
\(\ket{\Psi}\)~\cite{HaegemanGeometry2014,Holtz2012}.
The Hilbert-norm nearest-point map onto such a submanifold is a local
retraction: on a sufficiently small ambient neighborhood of \(\ket{\Psi}\),
every state has a unique nearest point in a small manifold neighborhood of
\(\ket{\Psi}\), and this map is \(C^2\)~\cite{AbsMal2012}.
Nearest-point projection in the Hilbert norm is Euclidean projection for the
real inner product \(\operatorname{Re}\langle\Phi|\Omega\rangle\).
The derivative of the retraction at \(\ket{\Psi}\) is therefore the
real-orthogonal projector onto \(T_{\Psi}\mathcal{M}_{\boldsymbol{\chi}}\).
Since that tangent space is complex-linear, the real-orthogonal projector
coincides with the complex Hilbert-space projector \(P_{\Psi}^{[1]}\).
The stated expansion is the \(C^2\) Taylor expansion of the retraction at
\(\ket{\Psi}\).
\end{proof}

\begin{theorem}[First-order agreement with ideal endpoint projection]
\label{thm:infinitesimal-endpoint-projection}
Let \(\ket{\Psi}\in\mathcal{M}_{\boldsymbol{\chi}}\) and let
\(H_g=H_g^\dagger\).  For sufficiently small real \(\epsilon\),
\begin{align}
    \Pi_{\mathcal{M}_{\boldsymbol{\chi}}}
    \left(
        e^{-\mathrm{i}\epsilon H_g}\ket{\Psi}
    \right)
    &=
    \ket{\Psi}
    -
    \mathrm{i}\epsilon
    P_{\Psi}^{[1]}
    H_g\ket{\Psi}
    +
    \mathcal{O}(\epsilon^2).
    \label{eq:endpoint-projection-expansion}
\end{align}
The fixed-rank TDVP solution obeys
\begin{equation}
    \ket{\Psi_g(\epsilon)}
    =
    \ket{\Psi}
    -
    \mathrm{i}\epsilon
    P_{\Psi}^{[1]}H_g\ket{\Psi}
    +
    \mathcal{O}(\epsilon^2).
    \label{eq:projected-gate-small-step-expansion}
\end{equation}
Thus the ideal projected endpoint and the fixed-rank TDVP state differ by
\(\mathcal{O}(\epsilon^2)\).
\end{theorem}

\begin{proof}
Exact evolution gives
\begin{equation*}
    e^{-\mathrm{i}\epsilon H_g}\ket{\Psi}
    =
    \ket{\Psi}
    -
    \mathrm{i}\epsilon H_g\ket{\Psi}
    +
    \mathcal{O}(\epsilon^2).
\end{equation*}
Applying \cref{lem:local-metric-projection} with
\(\ket{v}=-\mathrm{i}H_g\ket{\Psi}\), and using \(C^2\) regularity of
\(\Pi_{\mathcal{M}_{\boldsymbol{\chi}}}\) to absorb the
\(\mathcal{O}(\epsilon^2)\) remainder in the argument, gives
Eq.~\eqref{eq:endpoint-projection-expansion}.  Expanding
Eq.~\eqref{eq:projected-gate-flow} about \(\epsilon=0\) gives
Eq.~\eqref{eq:projected-gate-small-step-expansion}.  The tangent projector, and
hence the fixed-rank vector field, varies smoothly on
\(\mathcal{M}_{\boldsymbol{\chi}}\), which guarantees a unique local solution
and the stated expansion.
\end{proof}

The agreement in \cref{thm:infinitesimal-endpoint-projection} concerns the
derivative at zero gate strength, not equality of finite-gate endpoints.  Even
within the neighborhood where the local projection is defined, continuous
tangent projection updates the tangent space along its path, whereas endpoint
projection follows exact evolution and projects only the final state.  Their
endpoints can therefore differ on the curved MPS manifold, and the TDVP
endpoint can depend on the chosen generator path.

Practical compression routines are not assumed to realize the local
nearest-point projection.  Thus, \cref{thm:infinitesimal-endpoint-projection}
does not provide a first-order guarantee for TEBD, MPO truncation, or
variational fitting.

\begin{corollary}[Infidelity orders at small gate strength]
\label{cor:small-gate-infidelity-orders}
Under the assumptions of
\cref{thm:infinitesimal-endpoint-projection}, define
\begin{equation*}
    \ket{\Psi_{\mathrm{ex}}(\epsilon)}
    =
    e^{-\mathrm{i}\epsilon H_g}\ket{\Psi},
    \qquad
    \ket{\Psi_{\Pi}(\epsilon)}
    =
    \Pi_{\mathcal{M}_{\boldsymbol{\chi}}}
    \left(\ket{\Psi_{\mathrm{ex}}(\epsilon)}\right),
\end{equation*}
and write
\(\ket{\widehat{\Phi}}=\ket{\Phi}/\lVert\Phi\rVert\) for the normalization of
a non-zero state.  Then the
phase-insensitive infidelity between the fixed-rank TDVP state and the ideal
projected endpoint satisfies
\begin{equation*}
    1-
    \left|
        \left\langle
            \widehat{\Psi_{\Pi}}(\epsilon)
            \middle|
            \widehat{\Psi_g}(\epsilon)
        \right\rangle
    \right|^2
    =
    \mathcal{O}(\epsilon^4).
\end{equation*}
For the normal component
\begin{equation*}
    \ket{q}
    =
    \left(I-P_{\Psi}^{[1]}\right)H_g\ket{\Psi},
\end{equation*}
the exact and TDVP states instead obey
\begin{align}
    \ket{\Psi_{\mathrm{ex}}(\epsilon)}
    -
    \ket{\Psi_g(\epsilon)}
    &=
    -\mathrm{i}\epsilon\ket{q}
    +
    \mathcal{O}(\epsilon^2),
    \\
    1-
    \left|
        \left\langle
            \widehat{\Psi_{\mathrm{ex}}}(\epsilon)
            \middle|
            \widehat{\Psi_g}(\epsilon)
        \right\rangle
    \right|^2
    &=
    \frac{\epsilon^2}{\lVert\Psi\rVert^2}
    \lVert q\rVert^2
    +
    \mathcal{O}(\epsilon^3).
\end{align}
Consequently, if \(\ket{q}\neq0\), the exact-gate state-vector error has a
non-zero linear leading term and its phase-insensitive infidelity has a
non-zero quadratic leading term.
\end{corollary}

\begin{proof}
Theorem~\ref{thm:infinitesimal-endpoint-projection} gives
\(\lVert\ket{\Psi_{\Pi}(\epsilon)}-\ket{\Psi_g(\epsilon)}\rVert
=\mathcal{O}(\epsilon^2)\).  Since \(\ket{\Psi}\neq0\), normalization is
smooth in a neighborhood of \(\ket{\Psi}\), so the normalized states have the
same difference order.  For normalized \(\ket{x}\) and \(\ket{y}\),
\begin{equation*}
    1-|\langle x|y\rangle|^2
    \leq
    \lVert x-y\rVert^2,
\end{equation*}
because the difference between the right- and left-hand sides is
\(\lvert1-\langle x|y\rangle\rvert^2\).  This proves the
\(\mathcal{O}(\epsilon^4)\) result.

Subtracting the TDVP expansion in
Eq.~\eqref{eq:projected-gate-small-step-expansion} from the exact expansion in
the proof of \cref{thm:infinitesimal-endpoint-projection} gives the stated
state-vector difference.  Since \(P_{\Psi}^{[1]}\) is the orthogonal tangent
projector and \(P_{\Psi}^{[1]}\ket{\Psi}=\ket{\Psi}\),
\(\langle\Psi|q\rangle=0\).  Exact evolution and fixed-rank TDVP both preserve
the norm \(\lVert\Psi\rVert\).  For normalized \(\ket{x}\) and \(\ket{y}\),
\begin{equation*}
    1-|\langle x|y\rangle|^2
    =
    \left\lVert
        \left(I-\ket{y}\!\bra{y}\right)\ket{x}
    \right\rVert^2.
\end{equation*}
For the normalized exact and TDVP states,
\(\ket{y(0)}=\ket{\Psi}/\lVert\Psi\rVert\).  Thus
\(\langle\Psi|q\rangle=0\) removes the
\(\ket{y}\!\bra{y}\ket{q}\) term at leading order, and the vector inside the
norm is
\(-\mathrm{i}\epsilon\ket{q}/\lVert\Psi\rVert
+\mathcal{O}(\epsilon^2)\).  Squaring its norm gives the quadratic expansion.
\end{proof}

At smooth fixed rank and sufficiently small gate strength, we have shown that
TDVP and the ideal local nearest-point projection of the exact gate endpoint
have the same first-order change.  Their normalized phase-insensitive
infidelity is therefore \(\mathcal{O}(\epsilon^4)\).  By contrast, when
\(H_g\ket{\Psi}\) has a non-zero normal component, the exact endpoint and the
TDVP state differ at first order, giving a non-zero quadratic leading term in
their infidelity.  Quadratic infidelity relative to exact gate application
therefore does not by itself test the endpoint-projection relation.  We next
turn from this infinitesimal fixed-rank result to the finite rank-adaptive
update used in the calculations.

\section{Practical rank-adaptive \\ gate update}
\label{sec:practical-gate-update}

The results above characterize instantaneous projected directions.
Circuit simulation requires a finite update that can be applied to an MPS.
We obtain this update by applying the symmetric two-site TDVP projector-splitting
sweep~\cite{Haegeman_2016,Paeckel_2019} to the gate window \(W_g=[a,b]\),
using the locality identities established above. The following subsections
describe the windowed sweep and the errors introduced by its finite
implementation.

\subsection{Windowed sweep}
\label{sec:circuit-implementation}
\label{sec:adaptive-projector-splitting}

For the product generators in Eq.~\eqref{eq:product-gate-generator}, each gate
substep uses an exact bond-one MPO on \(W_g\); its boundary tensors contain
\(G_a^{(g)}\) and \(G_b^{(g)}\), while the interior sites carry identity
operators.
We divide the unit gate coordinate \(s\) into
\(n_{\mathrm{sub}}\in\mathbb{N}\), \(n_{\mathrm{sub}}\geq1\), equal steps
\begin{equation}
    \Delta s
    =
    \frac{1}{n_{\mathrm{sub}}}.
    \label{eq:gate-substep-size}
\end{equation}
The parameter \(n_{\mathrm{sub}}\) changes only the resolution of the chosen
path, not the gate or the physical circuit time.
In mixed-canonical form, let \(A_C(n:n+1)\) be the vectorized two-site center and
let \(H_{\mathrm{eff}}^{[2]}(n)\) be the generator represented in the current
left and right environments, which are updated as the center moves.  The
two-site update is
\begin{equation}
    A_C(n:n+1)
    \longleftarrow
    \exp\!\left[
        -\mathrm{i}\frac{\Delta s}{2}
        H_{\mathrm{eff}}^{[2]}(n)
    \right]
    A_C(n:n+1).
    \label{eq:adaptive-two-site-update}
\end{equation}
On the bond at the end of the window, this uses \(\Delta s\) in place of
\(\Delta s/2\).
Between two-site updates, the vectorized single-site center
\(A_C(n)\) evolves backward under its effective generator
\(H_{\mathrm{eff}}^{[1]}(n)\) according to
\begin{equation}
    A_C(n)
    \longleftarrow
    \exp\!\left[
        +\mathrm{i}\frac{\Delta s}{2}
        H_{\mathrm{eff}}^{[1]}(n)
    \right]
    A_C(n).
    \label{eq:adaptive-single-site-correction}
\end{equation}
We apply each local exponential to the
current center vector by an adaptive Lanczos iteration \cite{Lanczos_1950, HochbruckLubich1997}. Starting from that
vector, the iteration builds an orthonormal Krylov basis in which the effective
generator is tridiagonal.  The exponential is evaluated on this small
tridiagonal matrix and returned as a linear combination of the Krylov vectors.
The subspace is enlarged until an estimate of the exponential residual falls
below a chosen tolerance, or until a fixed basis size is reached.
The sweep proceeds as follows:
\begin{enumerate}
    \item For \(n=a,\ldots,b-2\), apply the two-site half-step on \(n,n+1\),
    split by an SVD so that site \(n\) is left-orthonormal, apply the backward
    single-site half-step at \(n+1\), and merge the center with site \(n+2\).

    \item On the bond at the end of the window \(b-1:b\), apply the full two-site step and split
    toward the left.

    \item For \(n=b-1,\ldots,a+1\), apply the backward single-site half-step at
    \(n\), merge the center with site \(n-1\), apply the two-site half-step on
    \(n-1,n\), and split toward the left.
\end{enumerate}
This symmetric order returns the orthogonality center to the left boundary.  It
contains
\(2(b-a)-1\) two-site exponentials and \(2(b-a-1)\) backward single-site
corrections.  When \(b=a+1\), both half-sweep loops are empty and only the update at the end of the window remains.

After each two-site update, the SVD keeps at least one singular value, discards
the largest tail whose unnormalized squared weight is below a threshold
\(\tau_{\mathrm{SVD}}>0\), and then imposes \(\chi_{\max}\geq1\).  It does not
pad the bond by retaining singular directions that are numerically zero.
Repeating the sweep for all \(n_{\mathrm{sub}}\) substeps completes the gate application.

\subsection{Cost and errors}
\label{sec:cost-error-decomposition}

Restricting the sweep to \(W_g\) makes the number of local exponentials and
SVDs linear in the gate span \(b-a\), rather than in the full chain length.  Their cost still depends on the local
dimension and the current bond dimensions, and preparing the boundary
environments can require traversal of the rest of the chain.  Runtime can therefore still
depend on system size \(N\).

The finite update also introduces errors that the instantaneous identities do
not control.  Each substep combines a symmetric projector splitting, a Lanczos
approximation of the local exponentials, and an SVD that may discard weight or
hit \(\chi_{\max}\).  For a smooth ODE with exact subflows, that splitting would
give local error \(\mathcal{O}(\Delta s^3)\) and accumulated error
\(\mathcal{O}(\Delta s^2)\) over a fixed
interval~\cite{Haegeman_2016,Paeckel_2019}.  Truncation and a changing rank
remove that guarantee.  Decreasing \(\Delta s\) also adds SVDs, so refinement at
a fixed retention rule does not isolate splitting error.

The instantaneous analysis does not determine how these errors combine in a
completed gate.  The next section, therefore, separates checks 
of the structural
identities from tests of the finite update.

\section{Results}
\label{sec:results}

The numerical study separates tests of the structural identities from
comparisons of finite-update accuracy and cost. We first check the theoretical
predictions, then compare the projected gate update with established methods
for individual gates and circuits. The shared protocol below keeps the
reference construction, truncation rules, and reported metrics consistent
across these tests.

\subsection{Shared numerical protocol}
\label{sec:numerical-protocol}

We start every method from the same physical MPS and apply the same ordered
gates unless stated otherwise.  For individual gates, we use untruncated
dense application as the reference.  For circuits, we also perform a dense evolution of the
discretized circuit, so the reported infidelities exclude
product-formula error relative to continuous Hamiltonian evolution.

We compare the projected gate update with TEBD+SWAP and a direct MPO baseline.
The two apply-and-compress baselines use the same direct two-site SVD update for
nearest-neighbor gates.  For long-range gates, TEBD+SWAP moves the endpoints together along the
shortest nearest-neighbor-SWAP route, applies and splits the gate, and reverses the
route.  The direct MPO baseline contracts the exact gate MPO with the full MPS
and makes one left-to-right truncated-SVD pass, whereas Projection uses the
windowed sweep of Sec.~\ref{sec:circuit-implementation}.  We contract
single-qubit gates directly.  In the chain benchmarks, Projection uses the
projected sweep for every two-qubit gate, while in the square-lattice benchmarks, it
uses the direct update for nearest-neighbor gates and the projected sweep for long-range
gates (so all methods use standard contraction for nearest-neighbor gates).

In each capped SVD, we retain at
least one singular value, discard an unnormalized squared tail below
\(\tau_{\mathrm{SVD}}=10^{-13}\), and then impose \(\chi_{\max}\).  We factor
gate MPOs with a hard singular-value cutoff of \(10^{-14}\) and approximate
effective exponentials with adaptive Lanczos--Krylov tolerance
\(\varepsilon_K=10^{-12}\).  We state any changed tolerance where it is used.
Because the methods perform different sequences of SVDs, these local rules do
not define a common cumulative truncation-error budget.

For non-zero reference and approximate states, we report the phase-insensitive
normalized infidelity
\begin{equation}
    1-F
    =
    1-
    \frac{
        \left|\langle\Psi_{\mathrm{ref}}|\Psi_{\mathrm{app}}\rangle\right|^2
    }{
        \langle\Psi_{\mathrm{ref}}|\Psi_{\mathrm{ref}}\rangle
        \langle\Psi_{\mathrm{app}}|\Psi_{\mathrm{app}}\rangle
    },
    \label{eq:normalized-infidelity}
\end{equation}

We implement all methods in MQT-YAQS~\cite{YAQS,wille_mqt2024} and run the timed calculations
single-threaded on an Intel i5-13600KF.  For the main methods, we
report medians of three repeats after one untimed warm-up; for the variational
MPO method, we report one observation.
We  first remove cap, truncation, and time-stepping effects and test the
structural identities directly.

\subsection{Structural checks}
\label{sec:structural-checks}

Before we examine finite-gate errors, we ask whether the implementation
reproduces the structure established in
Sec.~\ref{sec:gates-projected-evolution}.  The tests address three linked
questions: whether the gate window reproduces the full projected action,
whether the omitted exterior terms cancel nontrivially, and whether the
rank-one obstruction persists in a completed update.

We generate three normalized exact-rank \(N=8\) MPS from three random seeds,
with maximum bond dimension $\chi_{\max}=4$, and choose a fixed normalized generic
Hermitian two-site generator $H_g$.  We evaluate the dense projector actions without
a bond cap, SVD truncation, or TDVP substeps, so these approximations cannot
explain the residuals.  We test the obstruction separately by applying the
finite gate update to the four-qubit example of
Corollary~\ref{cor:rank-one-long-range-gate-freezing}.

For \(X_g=H_g\ket{\Psi}\), we define the locality residuals
\begin{equation}
    r_j
    =
    \frac{
        \left\lVert
            \left(P_\Psi^{[j]}
            -\mathcal{P}_{\Psi,W_g}^{[j]}\right)X_g
        \right\rVert
    }{
        \left\lVert P_\Psi^{[j]}X_g\right\rVert
    },
    \qquad
    j\in\{1,2\}.
    \label{eq:structural-locality-residuals}
\end{equation}
The residuals compare the full projector actions with the corresponding sums
on \(W_g=[a,b]\).  We also compare the nearest-neighbor two-site action with
\(X_g\) and pair each exterior term with the term predicted to cancel it.  By
recording the norm before cancellation, we distinguish a genuine pairwise
cancellation from an equality between zero terms.

\begin{table*}[t]
\caption{\textbf{Structural checks.} This table summarizes the results of several of the numerical benchmarks used to check the mathematical statements established in this work.}
\label{tab:structural}
\centering
\small
\begingroup
\setlength{\tabcolsep}{5pt}
\renewcommand{\arraystretch}{1.12}
\begin{tabular}{@{}
    >{\raggedright\arraybackslash}p{0.34\textwidth}
    >{\raggedright\arraybackslash}p{0.39\textwidth}
    >{\raggedleft\arraybackslash}p{0.20\textwidth}@{}}
\toprule
\textbf{Claim tested}
    & \textbf{What is compared}
    & \textbf{Outcome} \\
\midrule
\multicolumn{3}{@{}l}{\itshape Localization and exactness} \\
\addlinespace[0.15em]
Fixed-rank and two-site gate-window locality\newline
    {\footnotesize Theorem~\ref{thm:gate-window-locality}}
    & Full projector actions versus the corresponding \(W_g\) sums
    & \begin{tabular}[t]{@{}r@{}}
        \(\max r_1 = 8.40\times10^{-16}\) \\
        \(\max r_2 = 9.97\times10^{-16}\)
      \end{tabular} \\
Nearest-neighbor exactness\newline
    {\footnotesize Corollary~\ref{cor:nearest-neighbour-two-site-exactness}}
    & Nearest-neighbor residual \(r_{\mathrm{NN}}\)
    & \(\max r_{\mathrm{NN}} = 8.20\times10^{-16}\) \\
\addlinespace[0.45em]
\multicolumn{3}{@{}l}{\itshape Cancellation mechanism} \\
\addlinespace[0.15em]
Fixed-rank exterior pair\newline
    {\footnotesize Theorem~\ref{thm:gate-window-locality}}
    & Normalized term norm
      \(t=\lVert\mathrm{term}\rVert/\lVert X_g\rVert\) and paired residual
      \(r_{\mathrm{pair}}\)
    & \begin{tabular}[t]{@{}r@{}}
        \(\min t = 2.24\times10^{-1}\) \\
        \(\max r_{\mathrm{pair}} = 4.03\times10^{-16}\)
      \end{tabular} \\
Two-site exterior pair\newline
    {\footnotesize Theorem~\ref{thm:gate-window-locality}}
    & Normalized term norm
      \(t=\lVert\mathrm{term}\rVert/\lVert X_g\rVert\) and paired residual
      \(r_{\mathrm{pair}}\)
    & \begin{tabular}[t]{@{}r@{}}
        \(\min t = 4.64\times10^{-1}\) \\
        \(\max r_{\mathrm{pair}} = 7.30\times10^{-16}\)
      \end{tabular} \\
\addlinespace[0.45em]
\multicolumn{3}{@{}l}{\itshape Limitation} \\
\addlinespace[0.15em]
Long-range-gate obstruction at rank-one bonds\newline
    {\footnotesize Proposition~\ref{prop:rank-one-long-range-gate-projection} and
    Corollary~\ref{cor:rank-one-long-range-gate-freezing}}
    & Normalized projected-action norms
      \(\rho_j=\lVert P_\Psi^{[j]}X_g\rVert/\lVert X_g\rVert\)
    & \begin{tabular}[t]{@{}r@{}}
        \(\rho_1 = 0\) \\
        \(\rho_2 = 0\)
      \end{tabular} \\
\bottomrule
\end{tabular}
\endgroup
\end{table*}

The first two rows in Tab.~\ref{tab:structural} test whether restricting
the projector sums to the gate window loses any instantaneous action.
No loss is observed in these cases: the
fixed-rank and two-site locality residuals, as well as the nearest-neighbor
exactness residual, are all at numerical precision.  The next two rows explain
why.  The exterior terms are individually non-zero, but their paired residuals
fall to numerical precision.  The implementation therefore reproduces both the
window identity and the boundary cancellations on which it relies.

The final row tests the complementary limitation.  We see that both projected actions
vanish for the four-qubit rank-one example.  The complete 2TDVP update remains
stationary at both tested bond caps and substep counts: the long-range
\(R_{XX}(\pi/2)\) update leaves \(\ket{0,0,0,0}\) unchanged up to global phase,
with peak bond dimension one and final infidelity \(1-F=0.5\).  This links
the instantaneous obstruction to the implemented gate update, while remaining
a state-dependent counterexample rather than a general claim that long-range
gates stall.

\subsection{Individual gates}
\label{sec:single-gate-results}

The structural checks concern exact instantaneous claims. We now ask how the
completed projected gate update depends on gate strength, the available bond
dimension, and substep resolution.  We use three normalized random exact-rank
\(N=12\) MPS generated from three random seeds, each with maximum bond dimension
$8$, and apply a long-range gate between sites 2 and 11.  Unless stated
otherwise, we set the update cap to \(\chi_{\max}=8\), use
\(n_{\mathrm{sub}}=1\), and retain the shared tolerances from the protocol.
For long-range gates, we also use a variational-MPO control that applies the
gate as an MPO and fits a rank-capped MPS to the resulting state.

We test Pauli-rotation infidelity versus angle, CNOT infidelity versus bond cap,
and substep resolution for a weak \(R_{XX}\) rotation and CNOT.  For
the resolution test, we compare the states produced by \(n\) and \(2n\) substeps
through the phase-aligned distance
\begin{equation}
    D_n
    =
    \min_\phi
    \left\lVert
        \ket{\Psi_n}-e^{\mathrm{i}\phi}\ket{\Psi_{2n}}
    \right\rVert_2,
    \label{eq:successive-refinement-distance}
\end{equation}
where \(\ket{\Psi_n}\) denotes the projected gate update with
\(n=n_{\mathrm{sub}}\) substeps.
\cref{fig:single-gate-results} summarizes these tests.

\begin{figure*}[t]
    \centering
    \includegraphics[width=\textwidth]{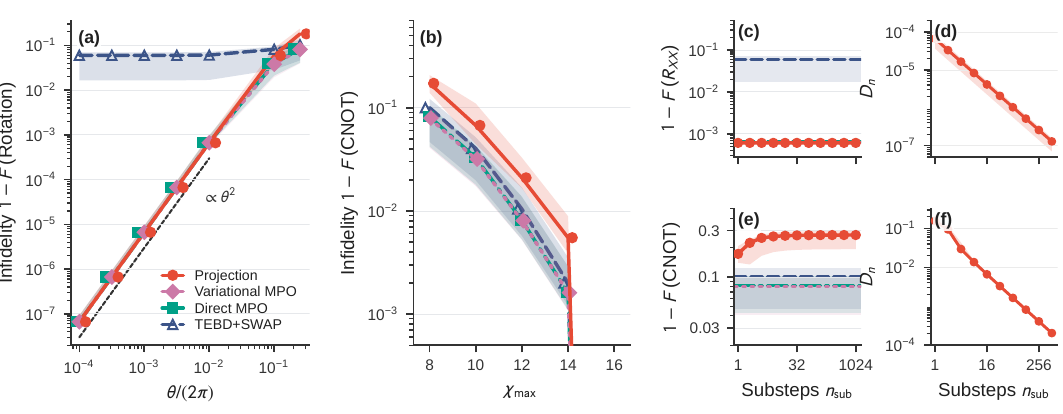}
    \caption{
        \textbf{Single-gate accuracy.}
        Long-range gates act between sites 2 and 11 of three \(N=12\) MPS;
        infidelity is relative to dense gate application.  Projection denotes
        the projected update and Direct MPO the direct MPO baseline.
        \textbf{(a)} Pauli rotations versus angle, with \(\theta^2\) guiding
        exact-gate error.  \textbf{(b)} CNOT versus bond cap.
        \textbf{(c)}, \textbf{(d)} Infidelity and phase-aligned distance \(D_n\)
        for \(R_{XX}\) at \(\theta/(2\pi)=10^{-2}\).
        \textbf{(e)}, \textbf{(f)} The same for CNOT.
        Curves show medians, and bands show the full ranges over nine cases in
        \textbf{(a)} and three states otherwise; overlapping markers are offset
        only for visibility.  Threshold truncation is suppressed in
        \textbf{(b)--(f)},
        except for Variational MPO in \textbf{(c)}.  Weak rotations give
        comparable infidelity across Projection and the MPO methods; refinement
        lowers \(D_n\) without lowering the CNOT exact-gate infidelity.
    }
    \label{fig:single-gate-results}
\end{figure*}

For weak rotations, we find comparable infidelity for the projected
update and both MPO methods.  The \(\theta^2\) line serves only as a guide
to error relative to the exact endpoint; it is neither a fit nor a test of~\cref{thm:infinitesimal-endpoint-projection}.  TEBD+SWAP retains a much higher error because it
executes and truncates the routing SWAPs even as the rotation approaches the
identity.

At the largest tested angles and \(\chi_{\max}=8\), both MPO methods give lower
infidelity than TEBD+SWAP and the projected update.  For CNOT, the
same ordering holds at the tested binding caps, but the gap disappears when all
four methods reach numerical precision at \(\chi_{\max}=16\).

Under higher substep resolution, \(D_n\) falls by orders of magnitude for both gates.  The infidelity remains comparable to the MPO methods for weak
\(R_{XX}\), whereas it rises and plateaus for CNOT.  Higher resolution therefore
stabilizes the successive projected update outputs without necessarily
improving fidelity.

In these tests, we find that the projected update gives infidelity comparable
to the MPO methods for weak rotations, while its CNOT accuracy depends on the
available bond dimension.  Increasing the substep resolution stabilizes
successive projected update outputs without guaranteeing closer agreement with
the exact gate.  We next test how these single-gate behaviors accumulate across
circuits.

\subsection{Circuit benchmarks}
\label{sec:circuit-results}

To test how the single-gate behavior accumulates, we compare length \(N=16\)
spin-\(1/2\) chains with \(4\times4\) square lattices, using open boundaries in
both.  Every two-qubit gate is adjacent in the chain ordering, whereas a
row-wise snake mapping of the square lattice also produces long-range gates.
This size keeps an untruncated dense reference feasible; we do not study
system-size scaling.  With \(X_i,Y_i,Z_i\) denoting Pauli matrices, the
circuits approximate the transverse-field Ising and XXX Heisenberg models,
respectively, are given by
\begin{align}
    H_{\mathrm I}
    &=-J\sum_{\langle i,j\rangle}Z_iZ_j-g\sum_i X_i,
    & J=g=1,
    \label{eq:benchmark-ising-hamiltonian}\\
    H_{\mathrm H}
    &=-J\sum_{\langle i,j\rangle}
        \left(X_iX_j+Y_iY_j+Z_iZ_j\right),
    & J=1.
    \label{eq:benchmark-heisenberg-hamiltonian}
\end{align}
Every method applies the same symmetric second-order product formula with
\(\Delta t=0.1\).  The Ising circuit uses \(R_X\) and \(R_{ZZ}\) gates, whereas
the Heisenberg circuit uses \(R_{XX}\), \(R_{YY}\), and \(R_{ZZ}\) gates.
Appendix~\ref{app:circuit-construction} gives the initial states, site map, and
gate ordering.  The displayed simulations use maximum bond dimension
\(\chi_{\max}=32\). In \cref{fig:circuit-infidelities}, Projection applies a
projected sweep to every two-qubit gate on the chains, extending the gate
window by one site on each side where available and using Lanczos tolerance
\(\varepsilon_K=10^{-12}\). On the square lattices, all methods use the common
TEBD update for MPS-nearest-neighbor gates. For MPS-long-range gates, Projection
sweeps the unpadded gate window at \(\varepsilon_K=10^{-5}\). All projected
sweeps use \(n_{\mathrm{sub}}=2\).

After each step, we record the stored MPS coefficient count
\(P=\sum_{i=1}^{N}d\,\chi_{i-1}\chi_i\).  This counts retained state
coefficients rather than total process memory.  We use \(\epsilon=10^{-2}\) as
the circuit tolerance and write \(n_\epsilon\) for the last Trotter step at which
\(1-F\leq\epsilon\) at that and every earlier sample.  A later return below the
tolerance does not restore reliability, and failure on the first step gives
\(n_\epsilon=0\).
\cref{fig:circuit-infidelities} compares the resulting fixed-cap trajectories.

\begin{figure*}[!t]
    \centering
    \includegraphics[width=\textwidth]{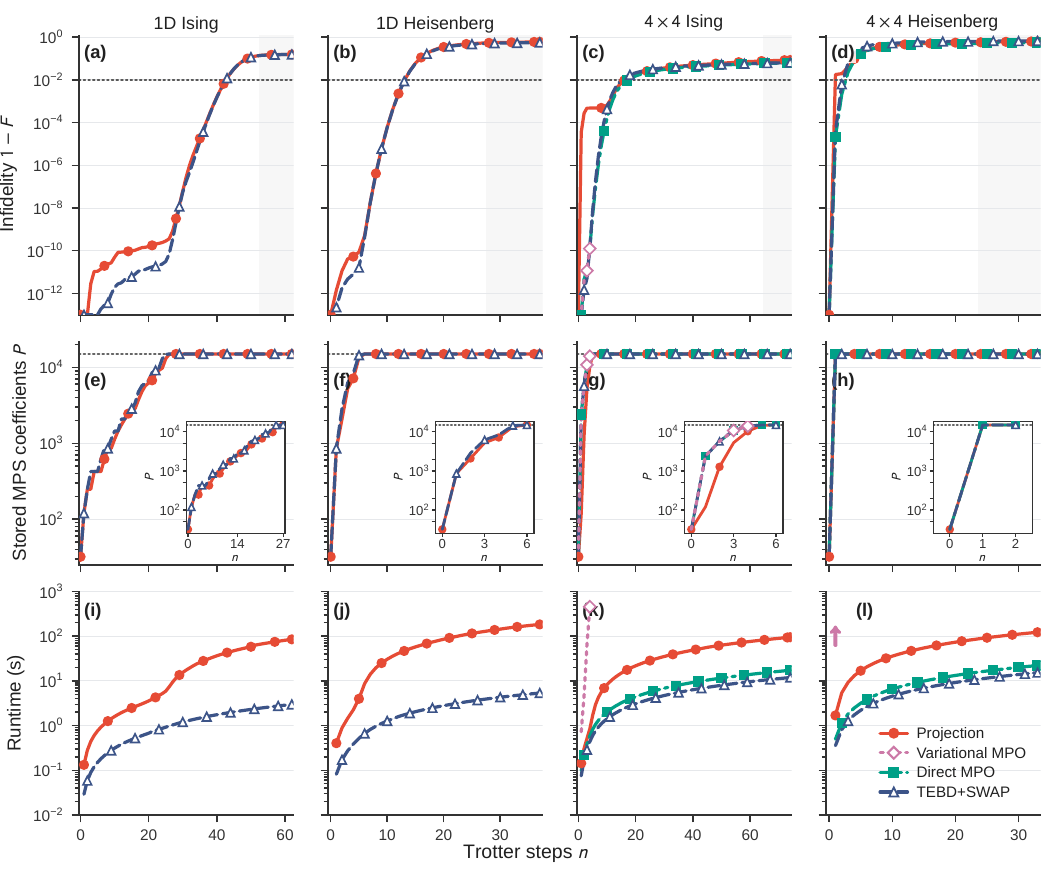}
    \caption{
        \textbf{Circuit evolution with bond dimension cap.}
        For four \(N=16\) circuits at \(\chi_{\max}=32\) and \(\Delta t=0.1\),
        \textbf{(a)--(d)} show infidelity, \textbf{(e)--(h)} stored MPS
        coefficients after each completed step, and \textbf{(i)--(l)} cumulative
        update time.  Infidelity is relative to dense execution of the same
        product-formula circuit.  Projection uses the projected sweep for every two-qubit chain gate,
        while on the square lattices all methods use TEBD for MPS-nearest-neighbor
        gates and Projection is reserved for MPS-long-range gates.  Direct MPO and
        Variational MPO are omitted in the chain columns because both reduce to TEBD+SWAP there.
        Dotted lines mark
        \(\epsilon=10^{-2}\) and the cap-filled profile. In \textbf{(a)--(d)}, shading marks the final ten-step display
        window rather than an asymptotic plateau.  Middle-row insets highlight early
        coefficient growth.  The upward caret in \textbf{(l)} is
        a \(>10^2\)-second lower bound for the incomplete Variational MPO step.
    }
    \label{fig:circuit-infidelities}
\end{figure*}

On the chains, Direct MPO, Variational MPO, and TEBD+SWAP coincide because every
two-qubit gate is nearest-neighbor, so their shared trajectory is shown alongside
Projection.  All four methods still
have the same reliability horizons, i.e., \(n_\epsilon=42\) for Ising and
\(n_\epsilon=13\) for Heisenberg.  The distinct Projection trajectory does not
contradict \cref{cor:nearest-neighbour-two-site-exactness}, which concerns the
instantaneous two-site projected action rather than the finite split and
truncated update.

The square-lattice circuits separate the three main methods more strongly.  For
Ising, \(n_\epsilon=16,17,15\) for the projected update, direct MPO, and
TEBD+SWAP, respectively.  For Heisenberg, the corresponding values are \(0,2,2\), so none of
the three main updates remains reliable beyond two steps.  Variational MPO has
essentially the same infidelity as the direct MPO baseline at the sampled
\(4\times4\) Ising endpoint, but is more than two orders of magnitude slower.
It does not complete
the first \(4\times4\) Heisenberg step within a cutoff of \(10^2\) seconds.

In three of the four circuits, Projection retains fewer coefficients during
part of the initial rank growth, but the difference disappears once the
methods reach the cap-filled profile.  All three methods fill the profile after
one step for \(4\times4\) Heisenberg.  Across these tests, we therefore find no persistent advantage in
reliability or retained storage.  On the square lattices, Projection takes
about \(5.5\)--\(8.0\) times as long as the two direct baselines at the displayed
endpoints.  Because these fixed-cap runs are not
compared under a common accuracy requirement, they do not define an
accuracy--runtime frontier.  We next compare the methods at a shared error
threshold.

\subsection{Accuracy and cost}
\label{sec:accuracy-cost-results}

The previous fixed-cap circuit benchmarks do not show how the methods compare
under a common accuracy requirement.  Here, we require each method to
remain below the same error threshold at every Trotter step through a fixed
horizon.  For the \(4\times4\) Ising circuit through \(n_\star=15\), we define
the maximum infidelity over this horizon,
\begin{equation}
    E_\star
    =
    \max_{0\leq k\leq n_\star}
    \bigl[1-F(k)\bigr]
    \label{eq:prefix-infidelity}
\end{equation}
and call a run reliable when \(E_\star\leq\epsilon=10^{-2}\).  The three main
methods use \(\chi_{\max}\in\{4,8,16,24,26,28,30,32\}\).  Every run completes
the same 15-step circuit, so runtime is an outcome rather than a stopping
condition.  Projection again uses \(n_{\mathrm{sub}}=2\) and
\(\varepsilon_K=10^{-5}\) only for long-range gates.

We define \(P_{\max}\) as the largest MPS coefficient count encountered during
a run, which, for the MPO methods, includes the uncompressed MPO--MPS target,
whose bonds can exceed \(\chi_{\max}\).  It excludes temporary numerical
workspaces and therefore does not measure total process memory.
\cref{fig:circuit-cap-sweep} visualizes these results and compares \(E_\star\), \(P_{\max}\), and runtime
over the tested maximum bond dimensions.
\begin{figure}[!t]
    \centering
    \includegraphics[width=\columnwidth]{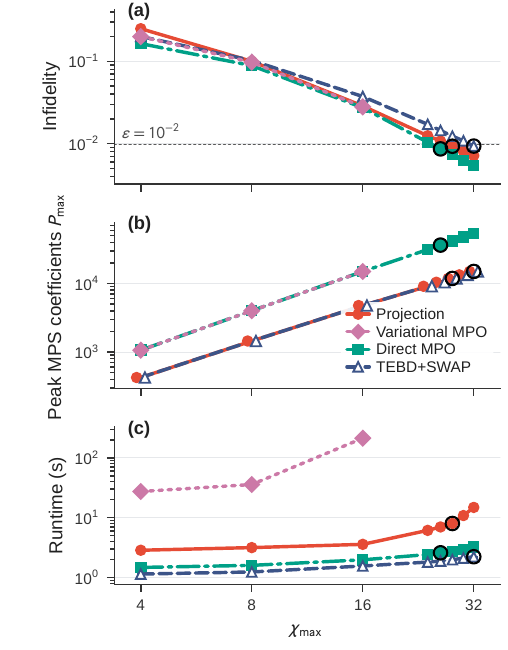}
    \caption{
        \textbf{Accuracy and cost at a fixed horizon.}
        For the \(4\times4\) Ising circuit through \(n_\star=15\),
        \textbf{(a)} shows the maximum infidelity \(E_\star\) over the fixed
        horizon, \textbf{(b)}
        peak MPS coefficient count \(P_{\max}\), and \textbf{(c)} runtime.
        Rings mark the first tested caps satisfying \(E_\star\leq10^{-2}\), i.e.,
        28, 26, and 32 for Projection, Direct MPO, and TEBD+SWAP.  For the MPO
        methods, \(P_{\max}\) includes the uncompressed target but excludes
        temporary numerical workspaces.  Projection applies nearest-neighbor gates
        directly and uses projected evolution only for long-range gates.  At
        the marked points, Projection uses fewer peak
        coefficients than Direct MPO and TEBD+SWAP but takes longer than both.
        Coincident markers in \textbf{(b)} are
        offset only for visibility, and lines are guides rather than scaling
        fits.
    }
    \label{fig:circuit-cap-sweep}
\end{figure}

The first tested caps that pass are 28, 26, and 32 for the projected update,
direct MPO baseline, and TEBD+SWAP, respectively.  Meeting the same threshold
does not imply identical achieved errors.  At these points, the projected
update lowers \(P_{\max}\) by 67\% and 21\% relative to the two baselines but
takes 3.1 and 3.6 times longer.  These are the first passing points on the
tested grid, not proven minima or global optima.

The lower peak count does not imply lower retained storage.  The completed
direct MPO state contains \(10\,456\) coefficients, compared with \(11\,880\)
for the projected update.  Relative to direct MPO, the projected update's lower
\(P_{\max}\) therefore reflects its avoidance of the uncompressed MPO--MPS
target.

Across \(\chi_{\max}=4,8,16\), Variational MPO never improves the maximum
infidelity over the fixed horizon relative to the direct MPO baseline and is
substantially slower than the three main methods at these caps.  All reported
fits converge.  This does not contradict the isolated-gate results because
each fitted state becomes the input to later gates.

Reclassifying the same cap sweep at \(\epsilon=2\times10^{-2}\) gives the same
first passing tested cap, \(\chi_{\max}=24\), for all three main methods.
At \(\epsilon=5\times10^{-3}\), none passes on the tested grid through 32.
The tested grid therefore
identifies a specific trade-off, not a general accuracy, storage, or runtime
hierarchy.

\section{Discussion}
\label{sec:discussion}

In this work, we have developed a geometric framework for quantum gate
application to matrix product states, establishing exact gate-window identities
and identifying how the available Schmidt supports enable or obstruct the
instantaneous gate action.
Applying quantum gates to an MPS under a bond cap requires balancing accuracy,
representation size, and runtime.  Established apply-and-compress methods make this
compromise after incorporating a discrete gate into the MPS.  We take a complementary view
and ask what can be learned by applying a gate through projected evolution
within the MPS variational space itself.  This leads to a gate-specific description of where the projected
evolution acts, whether the required variational directions are available, and how these structural
properties carry over to a full circuit update.

At the instantaneous level, the gate-window identity localizes the construction
of the projected action while retaining the influence of the surrounding MPS
through its boundary supports.  Those same supports can also exclude a required
direction, as the rank-one obstruction makes concrete.  At smooth fixed rank,
the comparison with ideal endpoint projection further shows that the two
approximation pictures agree infinitesimally while their finite updates remain
distinct.  Localization
and representability therefore provide design information, but neither alone
implies improved finite accuracy or lower cost.

Our numerical results carry this separation into the finite update.  The
structural checks reproduce the predicted identities and offer numerical validation, but
increasing the substep resolution can stabilize the projected trajectory
without necessarily moving it closer to the exact gate endpoint.  Across the circuit
benchmarks, the projected update also runs longer and does not consistently
improve fidelity or retained representation size.  In the regimes tested here,
localization does not translate by itself into a practical advantage, and the
results do not support the projected update as a general replacement for
apply-and-compress.

These limits make the next questions more specific.  One direction is to test
whether the boundary environment construction can be updated incrementally or
cached and reused across nearby windowed updates.  Another is whether we can use the represented
Schmidt supports and instantaneous projection residual to guide a choice among
a projected update, controlled support enlargement, and an apply-and-compress
update.  Controlled bond expansion or Krylov-subspace enrichment are natural candidates for that
enlargement.  A separate opportunity lies in the gate class.  The gate-level, finite update
benchmarks use standard two-qubit gates with product generators of
operator-Schmidt rank \(D_g=1\), while the corresponding nontrivial finite gates
generally have operator-Schmidt rank two.  The theory and structural checks also
cover higher-rank two-qubit generators.  Further benchmarking of such generators and
multiqubit gates would test how generator rank and gate arity affect the
accuracy and cost of the windowed update.  Such support-aware dispatch and
enrichment rules would define new algorithms and require independent
evaluation.
It is our hope that these structural insights will guide the development
of more effective tensor-network circuit simulators, in which the geometry
of the represented state informs the choice of gate updates and the
allocation of variational degrees of freedom.

\appendix
\section{Circuit construction}
\label{app:circuit-construction}

This appendix records the initial states, site map, and product-formula ordering
used for the circuit references and all approximate updates.  The Ising
circuits start from \(\ket{0}^{\otimes16}\).  The Heisenberg chain starts from
the N\'eel state, and the square lattice from the corresponding checkerboard
state.  Lattice coordinates \(r,c\in\{1,\ldots,4\}\) are mapped to the
one-dimensional MPS site
\begin{equation}
    j(r,c)=
    \begin{cases}
        4(r-1)+c, & r\ \text{odd},\\
        4(r-1)+5-c, & r\ \text{even},
    \end{cases}
    \label{eq:snake-map}
\end{equation}
so consecutive rows run in opposite directions.

The chain edges are partitioned into matchings with odd and even left
endpoints.  On the square lattice, horizontal edges are similarly partitioned
by the parity of their left column and vertical edges by the parity of their
upper row.  Matchings are ordered odd then even in one dimension and
horizontal-odd, horizontal-even, vertical-odd, vertical-even in two dimensions,
using these one-based parity labels.  Horizontal edges are ordered row-major,
vertical edges column-major, and chain edges by increasing left endpoint.  Let
\(\mathcal E\) be the concatenated edge list in this order.  The rotations
follow the convention
\begin{equation}
    R_P(\vartheta)=e^{-\mathrm{i}\vartheta P/2},
    \qquad
    R_{PP}(\vartheta)=e^{-\mathrm{i}\vartheta P\otimes P/2}.
    \label{eq:rotation-convention}
\end{equation}
In the order applied, one Ising step consists of
\(R_{ZZ}(-J\Delta t)\) on \(\mathcal E\),
\(R_X(-2g\Delta t)\) on sites \(1,\ldots,N\), and then the same bond gates on
the reversed list \(\mathcal E\).  For Heisenberg, form the ordered groups
\begin{equation}
    \begin{aligned}
        \mathcal G=\bigl[&
        (ZZ,\mathcal E_1),\ldots,(ZZ,\mathcal E_m),\\
        &(XX,\mathcal E_1),\ldots,(XX,\mathcal E_m),\\
        &(YY,\mathcal E_1),\ldots,(YY,\mathcal E_m)\bigr],
    \end{aligned}
    \label{eq:heisenberg-trotter-groups}
\end{equation}
where \(m=2\) for the chain or 4 for the square lattice.  All groups except
the last are applied with angle \(-J\Delta t\).  The last group is applied once
with angle \(-2J\Delta t\).  The preceding groups are then applied in reverse
group and edge order with angle \(-J\Delta t\).  Thus the half- and full-step
angles are \(-0.1\) and \(-0.2\), respectively.

\section*{Code and Data Availability}
All code and implementations are available in the open-source MQT-YAQS repository~\cite{YAQS,wille_mqt2024}. Data will be made available upon reasonable request.

\section*{Acknowledgments}
This work has been funded by the Deutsche Forschungsgemeinschaft (DFG, German Research Foundation) under Germany's Excellence Strategy--The Berlin Mathematics Research Center MATH+ (EXC-2046/1, project ID: 390685689),
the Cluster of Excellence ML4Q
(EXC 2004/1, project ID: 390534769)
and the CRC 183, the BMFTR (MuniQC-Atoms, QuSol, QPIC-1), the Munich Quantum Valley, Berlin Quantum, the Quantum Flagship (PasQuans2, Millenion), and the European Research Council (DebuQC). For MATH+, the Munich Quantum Valley, and Millenion, this is a joint project involving more than one project partner.
It has also been funded by
the European Union under the Horizon Europe Programme as part of the project NeQST (Grant Agreement 101080086) as well as by the European Research Council project DA QC (Grant Agreement 101001318).

\section*{Author Contributions Statement}
A.~S. conceived the project, developed the numerical method, designed and performed the numerical studies, analyzed and visualized the results, and led the writing of the manuscript. A.~S. and M.~F. jointly developed the theory, mathematical analysis, and derivations and implemented the algorithms. M.~A. provided
feedback on the mathematical analysis and derivations.  J.~E. advised on the
numerical study, interpretation, and presentation.  R.~M.~M. provided feedback
on implementation, debugging, and interpretation.  M.~E., M.~H., R.~W., and
C.~B.~M. supervised the work and provided feedback on the theoretical
development and manuscript.  All authors discussed the results and reviewed
the manuscript.

\section*{Competing Interests Statement}
The authors declare no competing interests.

\bibliography{bib.bib}

\end{document}